\documentclass[sigconf]{acmart}
\acmConference{Technical Report}{August 2024}{Porto, Portugal}
\usepackage{natbib}
\setcopyright{rightsretained}
\usepackage{xcolor} 

\newcommand{\equals}{\stackrel{\mathrm{def}}{=}}
\newcommand{\ignore}[1]{}   
   
\usepackage[linesnumbered,lined,commentsnumbered]{algorithm2e}

\begin{document}

\title{Efficiently Scheduling Parallel DAG Tasks on Identical Multiprocessors}
\titlenote{Version submitted in August 2024 to RTNS}

\author{Shardul Lendve}
\affiliation{%
  \institution{CISTER Research Centre and Faculdade de Engenharia da Universidade do Porto (FEUP)}
  \city{Porto}
  \country{Portugal}
}
\email{up202103275@edu.fe.up.pt} 

\author{Konstantinos Bletsas}
\affiliation{%
  \institution{CISTER Research Centre and ISEP}
  \city{Porto}
  \country{Portugal}
}
\email{ksbs@isep.ipp.pt}

\author{Pedro F. Souto}
\affiliation{%
  \institution{Faculdade de Engenharia da Universidade do Porto (FEUP) and CISTER Research Centre}
  \city{Porto}
  \country{Portugal}
}
\email{pfs@fe.up.pt}

\begin{abstract}

Parallel real-time embedded applications can be modelled as directed acyclic graphs (DAGs) whose nodes model subtasks and whose edges model precedence constraints among subtasks.
Efficiently scheduling such parallel tasks can be challenging in itself, particularly in hard real-time systems where it must be ensured offline that the deadlines of the parallel applications will be met at run time.
In this paper, we tackle the problem of scheduling DAG tasks on identical multiprocessor systems efficiently, in terms of processor utilisation.
We propose a new algorithm that attempts to use dedicated processor clusters for high-utilisation tasks, as in federated scheduling, but is also capable of reclaiming the processing capacity lost to fragmentation, by splitting the execution of parallel tasks over different existing clusters, in a manner inspired by semi-partitioned C=D scheduling (originally devised for non-parallel tasks).
In the experiments with synthetic DAG task sets, our Segmented-Flattened-and-Split scheduling approach achieves a significantly higher scheduling success ratio than federated scheduling.

\end{abstract}

\keywords{Directed Acyclic Graphs, identical multiprocessor systems, federated scheduling, semi-partitioned multiprocessor scheduling}
\maketitle

\section{Introduction}

Increasing demands for performance and energy efficiency underlie the popularity of commercial-off-the-shelf (COTS) multicore systems-on-a-chip (SoCs).
To best leverage the capabilities of such platforms, task parallelisation is needed. This involves dividing applications into subtasks
that can run in parallel on different cores. In the general case though, such subtasks can have precedence constraints and dependencies.
The Directed Acyclic Graph (DAG) model represents subtasks as nodes and dependencies among them as edges. DAG task scheduling is a complex problem, 
especially in hard real-time systems, where it must be ensured offline that the deadlines of the parallel applications will be met at run time. 

Federated scheduling~\cite{Li_CALGS_14} is a generalisation/counterpart of partitioned scheduling for parallel tasks~\cite{Baruah_15}. 
High-utilisation tasks that are not schedulable on a single processor are each assigned a dedicated \textit{cluster} of processors,
whereas tasks with utilisation not exceeding $1$ are scheduled as sequential (i.e., non-parallel) tasks on remaining processors.
This decomposes the scheduling problem into multiple simpler scheduling problems and also comes with proven performance guarantees, 
in the form of speed-up ratios~\cite{Li_CALGS_14}. However, it also suffers from some inefficiency due to fragmentation at the cluster level.

To usefully reclaim the processing capacity lost to fragmentation, we propose an efficient scheduling arrangement inspired by \textit{semi-partitioned} 
multiprocessor scheduling. Under our approach, some parallel DAG tasks migrate, at carefully selected static offsets relative to their release, 
from one cluster to another. This is analogous to how (a few) tasks migrate between processors under C=D~\cite{Burns_DWZ_11} -- arguably, 
the best-performing low-scheduling-overhead multiprocessor scheduling algorithm for non-parallel tasks~\cite{Santos_LB_13, Brandenburg_Gul_16}.

Next in this paper,
Section 2 discusses related work on the scheduling of parallel tasks and on semi-partitioning.
Section 3 introduces the task model under consideration. Section 4, introduces the basic building blocks of our approach.
Section 5 formulates the new scheduling algorithm. Section 6 offers an experimental evaluation, using synthetic workloads.
Finally, Section 7 provides conclusions.

\section{Related works}
Graham's work \citep{Graham_69} upper-bounds the makespan of a single DAG scheduled over $m$ identical processors, as a function of its overall workload and 
the length of its longest path, obliviously to its topology, as long as its scheduling is work-conserving and precedence constraints are respected. 
However, in the more general problem, there exist multiple such DAG tasks, with different periods.

Some existing approaches for that problem (e.g.,~\cite{Baruah_BMSW_12}) schedule DAG subtasks directly using classical algorithms 
(e.g., EDF or fixed priorities), while adapting the analysis and criteria to accommodate the attributes of DAGs. Building on~\cite{Baruah_BMSW_12}, 
Bonifaci et al.~\cite{Bonifaci_MSW_13} and Li et al.~\cite{Li_ALG_13} examined the overall timing characteristics of DAG tasks without delving 
into their internal structures,
proving a resource augmentation bound of $2-1/m$ for global EDF.

Other approaches transform the DAG into a set of independent sequential threads before scheduling. The task decomposition method by 
Saifullah et al.\cite{Saifullah_LALG_13} has a resource augmentation bound of $4$ and $5$, under global EDF and partitioned deadline 
monotonic scheduling, respectively. Lakshmanan et al.\cite{Lakshmanan_KR_10} examined real-time fork-join task scheduling on multiprocessors 
within OpenMP and introduced a sequential approach via a stretching algorithm to optimise execution. The study analyzes best-case and worst-case 
scenarios, leading to a preemptive fixed-priority scheduling algorithm with a resource augmentation bound of 3.42. 
Qamhieh et al.\cite{Qamhieh_GM_14} proposed a stretching algorithm to transform DAGs into independent sequential threads with offsets and deadlines. 
Their method executes fully-stretched master threads on dedicated processors and uses the remaining processors to schedule independent constrained-deadline 
threads, using either global EDF or global Deadline Monotonic (GDM). 
Cao and Bian~\cite{Cao_Bian_20} later improved on the assignment of internal offsets and deadlines.

Jing Li et al.~\cite{Li_CALGS_14} introduced federated scheduling, which is an extension of partitioned scheduling for parallel tasks. 
Under federated scheduling, DAG tasks with utilization above $1$ execute on their own dedicated clusters of processors; 
whereas, remaining lower-utilisation DAG tasks are scheduled as sequential threads on the remaining processors, using, e.g., 
partitioned or global scheduling, and EDF or fixed priorities.

Baruah~\cite{Baruah_15b} proposed federated scheduling for constrained-deadline sporadic tasks, with a speedup bound of $(2 - 1/m)$. 
The extension~\cite{Baruah_15} to arbitrary deadline tasks has a speedup bound of  $(3 - 1/m)$.
Using this strategy of federated scheduling and task decomposition, Bhuiyan et al.\cite{Bhuiyan_GSGX_18} focused
on creating an energy-efficient scheduling algorithm for DAG tasks. More recently, Fei Guan et al.~\cite{Guan_PQ_23} proposed another variant 
of federated scheduling, specifically optimised for DAG tasks whose deadlines exceed their period. It leverages the observation
that, for such tasks, it can sometimes be more efficient, in terms of processor usage, to use multiple smaller clusters for a given DAG 
(i.e., interleaved jobs thereof) than one bigger cluster (with enough processors to bring its response time below the respective period).

Coming from a slighly different motivation, Wasly and Pellizzoni~\cite{Wasly_Pellizzoni_19} proposed bundled gang scheduling. Gang tasks 
are parallel tasks whose threads must be executed and preempted at the exact same times as each other; bundled gang scheduling generalised 
on this model by having the number of threads of the gang task change, at different fixed points in its execution. The authors also suggest 
a technique for transforming a DAG task into an instance of the bundled gang task model, so that it can be scheduled under their approach.
In the case of DAGs with a high internal degree of parallelism, this can 
perform better than federated scheduling.

Conversely to Wasly's model, which prescribes the parallel execution of certain nodes (i.e., on different processors), under the 
model by Shi et al.~\cite{Shi_GUBC_24}, certain nodes must be executed on the same processor as each other (i.e., not in parallel, 
by definition), for reasons such as, e.g., cache locality. The authors provide schedulability analysis for this model under a form of list scheduling.

In a very different approach, Ahmed and Anderson~\cite{Ahmed_Anderson_22} use servers (one for every node of each DAG task), which are then 
scheduled under global EDF. In particular, their scheduling arrangement can be applied to systems where multiple instances (jobs) of a given DAG task 
may simultaneously be present, possibly with an upper bound on the number thereof.
\section{System model}

Consider a platform with $m$ identical processors $P_1, \ldots P_m$ and a set of independent sporadic DAG tasks $\tau =\lbrace G_1, \ldots,G_n \rbrace$. 
Each DAG task $G_i$ has an interarrival time $T_i$ and a deadline $D_i \leq T_i$ and consists of $n_i$ sequential subtasks ("nodes") with precedence constraints,
indicated by directed edges. Those nodes are denoted as $\tau_{i,1},\ldots,\tau_{i,n_i}$ and $C_{i,j}$ denotes the worst-case execution time (WCET) 
of node $\tau_{i,j}$. 
For each job of a DAG task, any node cannot start its execution until all its predecessors (i.e., nodes pointing to it via an edge) have completed theirs. 
At run time, different nodes of a given DAG task can execute in parallel with each other, as long as this does not violate any precedence constraints; 
however, a node cannot execute in parallel with itself on multiple processors. When convenient, we may also assume an imaginary "source" node $\tau_{i,0}$ 
with edges pointing to (true) nodes of $G_i$ with no (other) incoming edges. Likewise, it may be convenient to introduce an imaginary "sink" node, 
to which all nodes without other outgoing edges point. Those dummy nodes have a WCET of zero. 

The volume (or workload) of $G_i$ is denoted by $W_i$ and corresponds to the sum of its nodes' WCETs, i.e., $W_i = \sum_{\tau_{i,j} \in G_i} C_{i,j}$. 
A path's length is equal to the sum of the WCETs traversed by it. The length of the longest path of $G_i$ is denoted by $L_i$.

When an instance (job) of $G_i$ arrives at time instant $t$, its source node is immediately ready, and its other nodes become ready, as soon as 
all their predecessors have completed. All nodes of the current job of $G_i$ must have completed by time $t_i+D_i$
and its next job can arrive at any time after $t_i+D_i$, but not earlier. At any time, up to $n$ DAGs may be active (one from each DAG task in 
$\lbrace G_1, \ldots,G_n \rbrace$) and need to be scheduled on the $m$ processors. Preemption of a node and resumption of its
execution on another processor is allowed. The objective is for all DAG tasks to meet their deadlines.

\begin{table}[h]
\centering
\begin{tabular}{|c|p{5cm}|}
\hline
Symbol & Meaning \\
\hline
$P_1, \ldots, P_m$                        & the set of processors \\
$\tau = \lbrace G_1, \ldots ,G_n \rbrace$ & the set of DAG tasks \\
$T_i$                                     & interarrival time of $G_i$\\
$D_i$                                     & relative deadline of $G_i$\\
$L_i$                                     & the length of the longest path of $G_i$\\
$W_i$                                     & the total workload of $G_i$\\
$\tau_{i,1},\ldots\tau_{i,n_i}$           & the set of nodes of $G_i$\\
$C_{i,j}$                                 & the WCET of node $\tau_{i,j}$\\
\hline 
\end{tabular}
\caption{Notation used to describe the scheduling problem; not specific to our scheduling approach.}
\label{tab:symbol:table}
\end{table}

\section{The "building blocks" of our approach}

In this section we introduce some techniques that serve as building blocks for our DAG task scheduling approach.
These are:
\begin{itemize}
\item The \textit{segmentation} of a DAG task into disjoint sub-DAGs (called \textit{parallel segments}) that will be executed in a sequence.
\item The \textit{flattening} of a segmented DAG task into a fixed (relative to its arrival) schedule on a target number of processors.
\item The \textit{splitting} of a flattened schedule, inspired by C=D semi-partitioning~\cite{Burns_DWZ_12}, over disjoint processor clusters.
\end{itemize}

In the next section, we will then explain how these techniques fit together, to give rise to our multiprocessor real-time DAG task scheduling approach.

\subsection{Segmentation of DAG tasks}
\label{sect:segmentation}

Grouping of a DAG task's nodes into parallel segments (henceforth, just "segments") and executing those segments in a sequence
is akin to introducing \textit{additional} node precedence constraints to the DAG. This aims to make its scheduling more manageable. 

Let $S_{i,1},\ldots,S_{i,s_i}$ denote the segments ($s_i$ in total) of DAG $G_i$. 
For any $k$$\in$$\lbrace$$1$$,$$\ldots$$,$$s_i$$-$$1$$\rbrace$, 
we enforce the rule that all nodes in $S_{i,k}$ must have completed before the execution of any node in $S_{i,k+1}$ can start.

Let $\xi(\tau_{i,j)}$ denote the \textit{maximum} distance (in terms of hops) from the dummy source node of $G_i$ to $\tau_{i,j}$, over all paths. 

\begin{equation}
\tau_{i,j} \in S_{i,k} \Leftrightarrow \xi(\tau_{i,j)}=k
\label{eq:segment:formation}
\end{equation} 

The formation of segments according to Equation~\ref{eq:segment:formation} ensures that 
(i)~all the nodes belonging to the same segment may be executed in parallel to each other, without the possibility of any precedence constraint 
violation and 
(ii)~no precedence constraint can be violated if the segments of a given DAG are executed in order of ascending index. These properties are 
formalised by Lemma~\ref{lem:segments} below:

\begin{lemma}
Let a DAG task $G_i$ be decomposed into $s_i$ parallel segments according to Equation~\ref{eq:segment:formation}. 
If, $\forall k \in \lbrace 1, \ldots s_i-1 \rbrace$  no node in $S_{i,k+1}$ starts its execution before all nodes in $S_{i,k}$ have completed, 
then no precedence constraint is violated. Additionally, after all nodes in $S_{i,k}$ have completed, the nodes in $S_{i,k+1}$ can be executed 
in parallel, or concurrently to each other, without the possibility of any precedence  constraint violation arising as a result of that.
\label{lem:segments}
\end{lemma}

\begin{proof}
The proof is by induction. The first step ($k=1$) trivially holds, because the nodes in $S_{i,1}$ are all eligible for execution as soon as the 
DAG task arrives, and they can be executed, in parallel or concurrently, without this possibly causing any precedence constraint violation.

For the induction step let us assume that all nodes of some $S_{i,k}$ ($k \in \lbrace 1, \ldots s_i-1 \rbrace$) have completed and no precedence 
constraint violation has occured so far. We want to prove that the nodes of $S_{i,k+1}$ can then be executed, in parallel or concurrently to each 
other, without this possibly causing any precedence constraint violation.

The fact that all nodes of $S_{i,k}$ have completed implies that the nodes of all preceding segments, i.e., $S_{i,1}$ to $S_{i,k-1}$ have also 
completed, in that order -- because no segment's nodes becomes eligible for execution until its predecessor segment's nodes have all completed, 
according to the enforced scheduling rule. 
Consider a node $\tau_{i,j} \in S_{i,k+1}$. For any predecessor $\tau_{i,\ell}$ of $\tau_{i,j})$, it holds that $\xi(\tau_{i,\ell}) \leq k$; 
this follows from the fact that $\xi(\tau_{i,j})=k+1$. This, in turn, means that, for every predecessor $\tau_{i,\ell}$ of $\tau_{i,j}$
it holds that $\tau_{i,\ell} \in \lbrace  S_{i,1},\ldots,S_{i,k}\rbrace$. Therefore all predecessors of all nodes in $S_{i,k+1}$ 
have completed, when all nodes in $S_{i,k}$ have completed. Which means that all nodes in $S_{i,k+1}$ are eligible to start their execution, at that time, 
and they can be can be executed, in parallel or concurrently to each other, without the possibility of a precedence constraint violation.
\end{proof}

\subsection{Flattening of a DAG task}
\label{sect:flattening}

Consider a segment $S_{i,k}$ and a target number of processors $m_i$. Let $W(S_{i,k}) \equals \sum_{\tau_{i,j} \in S_{i,k}} C_{i,j}$.

Using McNaughton's wraparound algorithm \cite{McNaughton_59}, it is possible to generate a schedule for the nodes of $S_{i,k}$ over $m_i$ 
processors, with a length (makespan) of $\frac{W(S_{i,k})}{m_i}$, as long as $C_{i,j} \leq \frac{W(S_{i,k})}{m_i},~\forall \tau_{i,j} \in S_{i,k}$. 
Otherwise (i.e., if $\exists \tau_{i,j} \in S_{i,k}: C_{i,j} > W(S_{i,k})$), the length of the respective schedule cannot be less than 
$\max_{\tau_{i,j} \in S_{i,k}} C_{i,j}$, for the simple reason that at least one of the nodes would have to be executed in parallel with 
itself to make the schedule any shorter.

The pseudocode in Figure~\ref{fig:pseudocode:segment:flattening} outlines the generation of a schedule for a segment $S_{i,k}$, 
with time specified relative to the start of its execution. The flattened schedule for an entire DAG is obtained by the concatenation of
the respective flattened schedules for its segments, in order of execution.

     \begin{figure}[htbp]
     \begin{flushleft}
     \begin{minipage}[t]{\columnwidth}  
     \centering
     {\footnotesize
     \begin{tabular}{rp{10cm}}
      1.&\texttt{function flatten\_segment($S_{i,k}$, $m^{\prime}$)                                                    }\\
      2.&\texttt{$\lbrace$//without loss of generality, over processors $P_1$ to $P_{m^{\prime}}$                      }\\
      3.&\texttt{\ schedule=new Schedule();                                                                            }\\
      4.&\texttt{                                                                                                      }\\ 
      5.&\texttt{\ int Cmax = $\max_{\tau_{i,j} \in S_{i,k}}$ $C_{i,j}$;                                               }\\   
      6.&\texttt{\ int W    = $\sum_{\tau_{i,j} \in S_{i,k}}$ $C_{i,j}$;                                               }\\ 
      7.&\texttt{                                                                                                      }\\ 
      8.&\texttt{\ \textbf{if} (W/$m^{\prime}$ > Cmax)                                                                 }\\ 
      9.&\texttt{\ \ schedule.length = ceil(W/$m^{\prime}$);                                                           }\\ 
     10.&\texttt{\ \textbf{else}                                                                                       }\\ 
     11.&\texttt{\ \ schedule.length = Cmax;                                                                           }\\ 
     12.&\texttt{                                                                                                      }\\ 
     13.&\texttt{\ int O=0;                                                                                            }\\ 
     14.&\texttt{\ int p=1;                                                                                            }\\ 
     15.&\texttt{                                                                                                      }\\              
     16.&\texttt{\ \textbf{for} (${\tau_{i,j} \in S_{i,k}}$)                                                           }\\ 
     17.&\texttt{\ $\lbrace$start = O;                                                                                 }\\
     18.&\texttt{\ \ end   = (O+$C_{i,j}$) modulo schedule.length;                                                     }\\
     19.&\texttt{                                                                                                      }\\
     20.&\texttt{\ \ \textbf{if} (end>start) // no wrap-around                                                         }\\
     21.&\texttt{\ \ $\lbrace$schedule.add\_interval($\tau_{i,j}$,$P_p$,start,end);//\tiny{node, processor, start, end}}\\
     22.&\texttt{\ \ \  O=end;                                                                                         }\\
     23.&\texttt{\ \ \  if (O==schedule.length) //corner case                                                          }\\
     24.&\texttt{\ \ \  $\lbrace$O=0;                                                                                  }\\
     25.&\texttt{\ \ \  \ p=p+1;                                                                                       }\\
     26.&\texttt{\ \ \  $\rbrace$                                                                                      }\\
     27.&\texttt{\ \ $\rbrace$                                                                                         }\\ 
     28.&\texttt{\ \ \textbf{else}                                                                                     }\\ 
     29.&\texttt{\ \ $\lbrace$schedule.add\_interval($\tau_{i,j}$, $P_p$, start, schedule.length);                     }\\              
     30.&\texttt{\ \ \ schedule.add\_interval($\tau_{i,j}$, $P_{p+1}$, 0, end);                                        }\\ 
     31.&\texttt{\ \ \ O=end;                                                                                          }\\
     32.&\texttt{\ \ \ p=p+1;                                                                                          }\\
     33.&\texttt{\ \ $\rbrace$                                                                                         }\\
     34.&\texttt{\ $\rbrace$                                                                                           }\\
     35.&\texttt{                                                                                                      }\\
     36.&\texttt{\ \textbf{return} schedule;                                                                           }\\
     37.&\texttt{$\rbrace$                                                                                             }\\
     \end{tabular}
     } 
     \caption{Pseudocode: flattening of an individual segment}
     \label{fig:pseudocode:segment:flattening}
     \end{minipage}
     \end{flushleft}
     \end{figure}

\begin{figure}[htbp]
\centerline{\includegraphics[width=7cm]{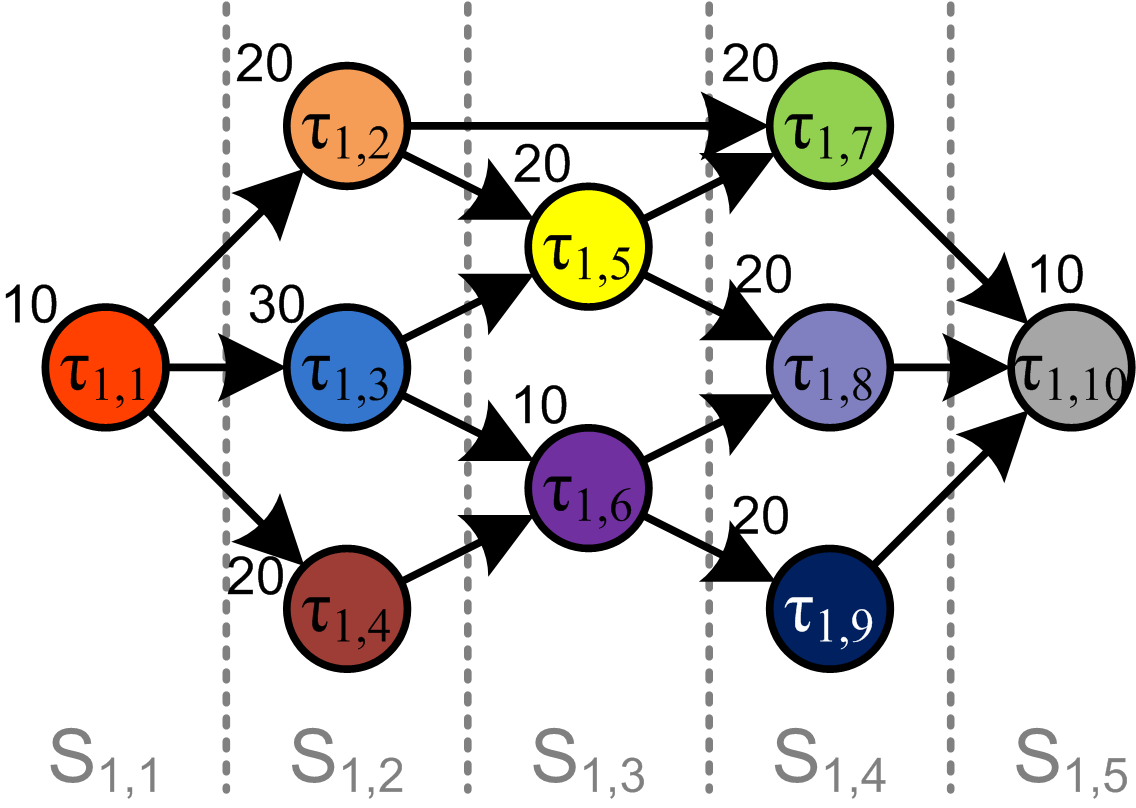}}
\caption{A DAG task $G_1 $, with a deadline of $130$, that consists of 10 subtasks (nodes). 
Each node's WCET is shown to its upper left. The DAG is broken into 5 parallel segments, separated by vertical gray dashed lines.}
\label{fig:dagexample}
\end{figure}

\begin{figure}[htbp]
\centerline{\includegraphics[width=8cm]{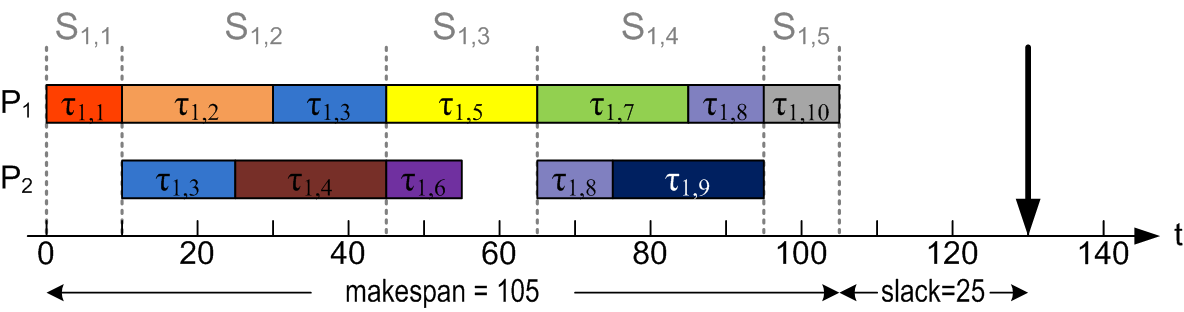}}
\caption{Flattening of the DAG of Fig.~\ref{fig:dagexample} over 2 processors. The resulting makespan is 105.}
\label{fig:sub2}
\end{figure}

To better illustrate the segmentation and flattening processes, consider the example of Figure~\ref{fig:dagexample}, where a 
DAG ${G_1}$ and its segmentation according to Equation~\ref{eq:segment:formation} are shown. Figure~\ref{fig:sub2} 
depicts a schedule for that same DAG on 2 processors, obtained via flattening each of its segments and concatenation of the 
respective schedules.

The algorithm of Figure~\ref{fig:pseudocode:segment:flattening} can be used in a loop to find the smallest feasible number of processors for a flattened DAG task,
and the corresponding schedule, as in Figure~\ref{fig:pseudocode:segment:maximum:flattening}. Lines 2 to 6 in the pseudocode of 
Figure~\ref{fig:pseudocode:segment:maximum:flattening} check if it is feasible to flatten the input DAG.
Observe that a segment's schedule's length is lower-bounded by the greatest node WCET in the segment, irrespective of the number of available processors.

     \begin{figure}[htbp]
     \begin{flushright}
     \begin{minipage}[t]{\columnwidth}  
     \centering
     \footnotesize{
     \begin{tabular}{rp{10cm}}
      1.&\texttt{function feasibly\_max\_flatten($G_i$)                                                               }\\
      2.&\texttt{$\lbrace$int LB=0;                                                                                   }\\
      3.&\texttt{\ for ($S_k$ $\in$ segments\_of($G_i$))                                                              }\\
      4.&\texttt{\ \ LB=LB+$\max_{\tau_{i,j} \in S_{i,k}}$ $C_{i,j}$;                                                 }\\ 
      5.&\texttt{\ if (LB<=$D_i$)                                                                                     }\\   
      6.&\texttt{\ \ return FAILURE;                                                                                  }\\ 
      7.&\texttt{                                                                                                     }\\ 
      8.&\texttt{\ int $m^{\prime}$ = ceil($W_i$ / min($D_i$, $T_i$)); //initialisation                               }\\ 
      9.&\texttt{\ while (true)                                                                                       }\\
     10.&\texttt{\ $\lbrace$schedule = new Schedule();                                                                }\\
     11.&\texttt{\ \ for ($S_{i,k}$ $\in$ segments\_of($G_i$))                                                        }\\ 
     12.&\texttt{\ \ \ schedule = schedule.append(flatten\_segment($S_{i,k}$, $m^{\prime}$);                          }\\   
     13.&\texttt{\ \ if (schedule.length <= $D_i$)                                                                    }\\ 
     14.&\texttt{\ \ \ return schedule;                                                                               }\\ 
     15.&\texttt{\ \ else                                                                                             }\\ 
     16.&\texttt{\ \ \ $m^{\prime}$ = $m^{\prime}$+1;                                                                 }\\ 
     17.&\texttt{\ $\rbrace$                                                                                          }\\ 
     18.&\texttt{$\rbrace$                                                                                            }\\ 
     \end{tabular}
     } 
     \caption{Pseudocode for flattening a DAG task over the smallest feasible number of processors}
     \label{fig:pseudocode:segment:maximum:flattening}
     \end{minipage}
     \end{flushright}
     \end{figure}

If the algorithm of Figure~\ref{fig:pseudocode:segment:maximum:flattening} fails, it is always worth considering Graham's bound~\cite{Graham_69} 
for the given DAG, as an alternative. The latter upper-bounds the makespan of a DAG task, scheduled by any work-conserving node scheduling algorithm on $m^{\prime}$ processors as

\begin{equation}
M_i = L_i + \left\lceil \frac{W_i-L_i}{m^{\prime}}\right\rceil
\label{eq:graham}
\end{equation}
where $L_i$ is the length of the DAG's longest path, $W_i$ is its volume.
In very rare cases, such as the example of Figure~\ref{fig:graham:counterexample}, 
the above bound can be shorter than the corresponding flattened schedule.

It follows from Graham's bound (see Theorem~2 in~\cite{Li_CALGS_14}) that an upper bound for the smallest sufficient dedicated cluster size for a 
constrained-deadline\footnote{That result was formulated for implicit-deadline DAG tasks, but it equally holds for constrained-deadline DAGs.} 
DAG task to be schedulable under any work-conserving node scheduling algorithm is given by

\begin{equation}
m^{\prime} = \left\lceil \frac{W_i-L_i}{D_i-L_i}\right\rceil
\label{eq:graham:cluster:size}
\end{equation}

\begin{figure}
\centerline{\includegraphics[width=5cm]{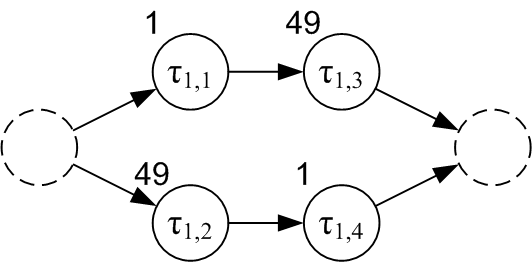}}
\caption{Graham's bound for the makespan of the above DAG over 2 processors is 50+(100-50)/2=75. The length of the corresponding 
segmented-and-flattened schedule is 49+49=98.}
\label{fig:graham:counterexample}
\end{figure}

\subsection{Splitting of a DAG task's execution over multiple clusters}

Assigning heavy DAG tasks to dedicated clusters can lead to fragmentation, analogous to bin-packing fragmentation.
The cumulative unused processor capacity over all clusters may be substantial, but not usable. To remedy this, we devise 
an arrangement inspired by \textit{semi-partitioning} -- an efficient solution to bin-packing fragmentation in the multiprocessor scheduling of non-parallel tasks. 
We specifically draw from \textit{C=D scheduling}~\cite{Burns_DWZ_12}.

\textbf{Encapsulation of a flattened schedule by a gang task:}
Consider a cluster of $m^\prime > 1$ processors where some heavy DAG task $G_{g}$ has already been assigned, to execute under its flattened schedule 
with a makespan of $M^\prime$. If a job by $G_{g}$ arrives at time $t$, not all $m^\prime$ processors will necessarily be continuously busy during the 
interval $[t,t+M^\prime)$, as flattened schedules are not work-conserving in the general case. Still, it is convenient, as we will argue, to encapsulate 
the execution of $G_{g}$ into a \textit{gang task} with a WCET of $C_{g}=M^\prime$ and a degree of parallelism equal to $m^\prime$. All parallel threads 
of a gang task must be executed in lockstep, i.e., released/preempted/resumed at the same time instants as each other. Enforcing such gang semantics allows 
us to run additional, potentially interfering, workload on the cluster without potential misalignment of the schedules of $G_{g}$ on the different processors 
(leading, e.g., to potential violations of node precedence constraints or scheduling of the same "wrapped-around" node on two processors simultaneously). 
Rather, the flattened schedule of $G_{g}$, encapsulated by the gang task, can be preempted, on all processors, for up to $D_{g}-C_{g}$ time units in total,
with no detriment to the schedulability of $G_{g}$, in order to execute other workload.

\textbf{Reducing gang scheduling to uniprocessor scheduling:}
Gang scheduling can be inefficient, because it is inherently prone to priority inversions (e.g. 2-D blocking).
For example, a ready middle-priority task may be unable to execute because it requires more processors than those left unclaimed by higher-priority tasks;
conversely, a lower-priority task with smaller degree of parallelism may execute instead in its place. However, such priority inversion cannot occur if all gang tasks 
have the same degree of parallelism, equal to the number of available processors. The scheduling problem then becomes equivalent to uniprocessor scheduling, 
because, at any instant, all the processors in a cluster will be executing the same task as each other. In turn, this makes it easier to split the execution of 
parallel DAG tasks over different processor clusters, analogously as to how the execution of non-parallel tasks is split across different processors 
in semi-partitioned scheduling.

\textbf{Results from uniprocessor and semi-partitioned multiprocessors scheduling of non-parallel tasks}

C=D~\cite{Burns_DWZ_12} is based on partitioned EDF. It uses a uniprocessor EDF scheduler on each processor but it splits the execution of some tasks across multiple 
processors, in a carefully managed manner, when those tasks cannot feasibly be assigned entirely to any processor. The "pieces" of such split tasks are modelled as separate tasks,
partitioned to different processors, and executing in a "pipelined" manner. 
By reporting to the EDF scheduler, purely for scheduling decision purposes, for each piece of a split task (except possibly the last one), 
a deadline equal to its execution time budget on the respective processor, it is ensured that such a zero-laxity piece always executes at the 
highest priority, and is never preempted. 
Hence the name "C=D". Subject to that, a split task's execution budget on a given processor is set, using sensitivity analysis, to the maximum value that preserves 
the schedulability of other tasks on that processor. Specifically: 

Let $C_s$ denote the execution budget of a non-final piece of some split task on processor $P_p$, with $D_s = C_s$. Then, using the exact schedulability test for 
uniprocessor EDF schedulability~\cite{Baruah_MR_90}, $C_s$ can be set to the maximum value that satisfies
 
\begin{align}
\mathrm{dbf}(\tau_s, t) + \sum_{\tau_i \in \Gamma_p} \mathrm{dbf}(\tau_i, t) \leq t, \forall t >0 
\label{eq:sensitivity:analysis:dbf}
\end{align}
where $\Gamma_p$ is the set of other tasks assigned to $P_p$ and the \textit{demand-bound function} of a periodic or sporadic task $\tau_i$ is defined as 

\begin{equation}
\mathrm{dbf}(\tau_i, t) \equals \max \left( 0, \left\lfloor \frac{t-D_i}{T_i} \right\rfloor +1 \right) C_i
\label{eq:dbf}
\end{equation}

The above sensitivity analysis problem can be enormously sped up using the QPA technique~\cite{Zhang_Burns_09, Zhang_BB_11}, but it can still 
have long running time. Fortunately, if all tasks in $\Gamma_p$ have implicit deadlines, 
Santos-Jr's test, known from~\cite[Theorem 3]{Santos_Lima_12} and~\cite[Theorem 3]{Santos_LB_14} (and henceforth referred to as
\textit{Augusto's test}), can be used instead.
It is a simpler, sufficient test that performs nearly as well in practice:

\begin{equation}
\frac{C_s}{T_s}  \leq  \frac{1-\displaystyle\sum_{\tau_i \in \Gamma_p} U_i}{1 +\frac{\displaystyle\sum_{\tau_i \in \Gamma_p} U_i}{\left\lfloor\frac{\min_{\tau_i \in \Gamma_p} T_i}{T_s}\right\rfloor}}
\label{eq:augusto:3}
\end{equation}
where $U_i \equals \frac{C_i}{T_i}$.

Because uniprocessor EDF is a sustainable scheduling algorithm \cite{Baruah_Burns_06}, Augusto's test can be made to also work, at the cost of some additional pessimism, 
when some of the tasks in $\Gamma_p$ have constrained deadlines, if the deadlines ($D_i$) of such tasks are used in Inequality~\ref{eq:augusto:3}, 
in place of of their interarrival times. For which case, it can be expressed as

\begin{equation}
\frac{C_s}{T_s} \leq \
\frac{1-\displaystyle\sum_{\tau_i \in \Gamma_p} \delta_i}{1+\frac{\displaystyle\sum_{\tau_i \in \Gamma_p} \delta_i}{\left\lfloor \frac{\min_{\tau_i \in \Gamma_p} \displaystyle D_i}{\displaystyle T_s}\right\rfloor}}
\label{eq:augusto:3:constrained}
\end{equation}
where $\delta_i \equals \frac{C_i}{\min(D_i, T_i)}$.

\textbf{Sensitivity analysis for the splitting of the execution of a parallel DAG task}

Consider a cluster $Q_q$ of $m^{\prime}$ processors, where a set $\Gamma(Q_q)$ of DAG tasks has already been assigned. Each parallel DAG task 
$G_i \in \Gamma(Q_q)$ can be modelled by a gang task $\tau_i$ with:
\begin{itemize}
\item a degree of parallelism equal to $m^{\prime}$;
\item the same deadline ($D_i$) and interrarival time ($T_i$)as $G_i$; and 
\item a WCET $C_i$ equal to the makespan $M_i$ of the flattened schedule for $G_i$ on $m^{\prime}$ processors -- or, if that flattened schedule was infeasible,
an upper bound on its makespan, if it were executing on $m^{\prime}$ dedicated processors under any work-conserving algorithm (as derived, e.g. by Equation~\ref{eq:graham}).
\end{itemize}

Then, because all tasks have the same degree of parallelism $m^{\prime}$, matching the number of processors in $Q_q$, the cluster $Q_q$ is schedulable under gang EDF if 
and only if the equivalent non-gang task system is schedulable under uniprocessor EDF.

Now consider a different DAG task $G_j$, encapsulated by a gang task $\tau_s$ with the same degree of parallelism ($m^{\prime}$), 
which is to be scheduled on the same cluster at the highest priority. The tasks in $\Gamma(Q_q)$ are schedulable if and only if the following 
uniprocessor task set is schedulable under EDF:

\begin{itemize}
\item $\tau_s \equals (C_s, D_s=C_s, T_s)$
\item $\tau_i \equals (C_i, D_i, T_i)$,~$\forall i: G_i \in \Gamma(Q_q)$
\end{itemize}

Clearly the sensitivity analysis for C=D can be used to determine the maximum value for $C_s$ such that the above system is feasible.

\textbf{Example of split execution of a DAG on multiple clusters}

Consider again the same DAG $G_1$ as in Figure~\ref{fig:dagexample}. There are 5 processors, and there already exist two clusters, $Q_1=(P_1, P_2)$ and $Q_2=(P_3, P_4, P_5)$, 
each already with a DAG assigned to them. The flattened schedule of $G_1$ on two processors has a length of $105$ (Figure~\ref{fig:sub2}) but, according to sensitivity analysis, 
if $G_1$ is assigned to $Q_1$ at top priority, the threshold for feasibility is $60$ time units. We therefore encapsulate $G_1$'s execution on $Q_1$ by a 
sporadic gang task with a WCET of $60$ and a deadline of $60$. When that gang task is scheduled (under EDF), the first 60 time units of the flattened schedule (on two processors) 
of $G_1$ are executed by its threads -- see top of Figure~\ref{fig:splitting}. What remains to execute from $G_1$ after that point, would need to be executed on $Q_2$ - however
$Q_2$ has three processors, so it would be wasteful to use the remaining flattened schedule for $G_1$ on two processors (crossed out in Figure~\ref{fig:splitting}) and
have one of three gang threads for $G_1$ on $Q_2$ be idle. Instead, from the segmented representation of $G_1$, we subtract the amount execution that nodes would have already 
received (i.e., in $Q_1$). The segmented representation of the "rump" DAG that remains is shown in the bottom left of Figure~\ref{fig:splitting}). Nodes $\tau_1$ to $\tau_4$ and
$\tau_6$ (to be entirely executed in $Q_1$) are entirely missing and only a portion (of the execution time) of $\tau_5$ remains. The corresponding flattened schedule on three processors for the rump segment sequence has a length of 35 time units. This schedule is executed by three corresponding gang threads on $Q_2$ with a WCET of 35, that arrive
on $Q_2$ as soon as $G_1$'s execution on $Q_1$ stops (i.e., 60 time units after the arrival of $G_1$) (Figure~\ref{fig:splitting}), bottom right). That gang EDF task has a relative 
deadline equal to $D_1-60=130-60=70$ and, in our example, it is feasible to assign to $Q_2$.

\begin{figure}
\centerline{\includegraphics[width=8.5cm]{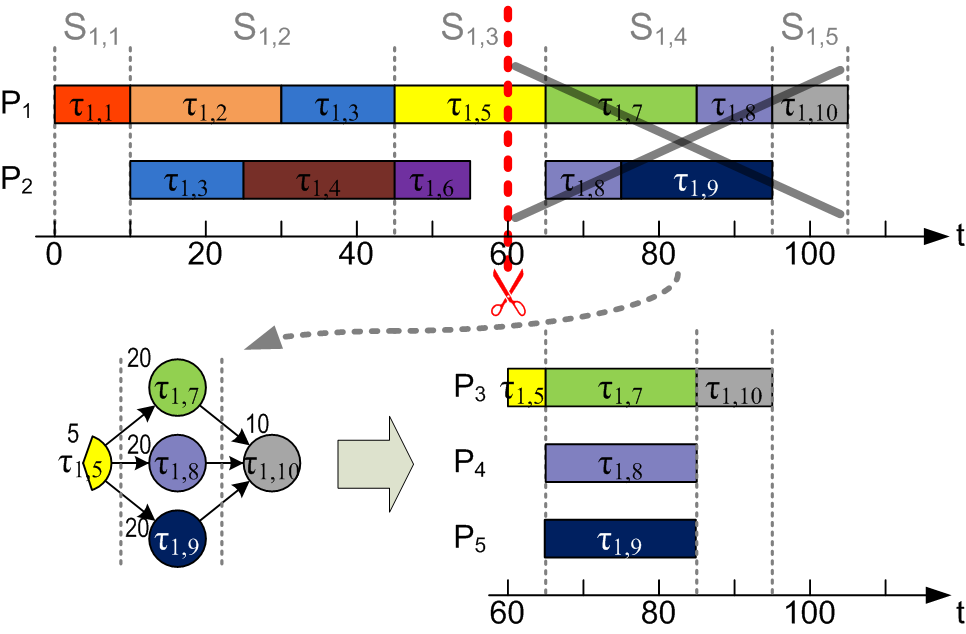}}
\caption{Illustration of semi-partitioned scheduling of a DAG over two clusters that consist of a different number of processors.}
\label{fig:splitting}
\end{figure}

\section{The complete algorithm}

In this section, we describe the complete algorithm, whose design makes use of the building blocks described in Section~4. 
The algorithm is presented in high-level pseudocode in Figure~\ref{fig:pseudocode:small}, to better convey its flow.
It is also presented in more detailed pseucodode (Figures~\ref{fig:pseudocode:detailed:part1} and~\ref{fig:pseudocode:detailed:part2})
in order to more rigorously document its workings. It performs two passes for DAG task assignment. 

\subsection{The two passes of the algorithm explained}

     \begin{figure}[htb]
     \begin{flushleft}
     \begin{minipage}[t]{8.5cm}  
     \centering
     \footnotesize{
     \begin{tabular}{rp{8cm}}
      1.&\texttt{//1st pass -- akin to federated scheduling                                                                       }\\
      2.&\texttt{\textbf{for} (every DAG $G_i$)                                                                                   }\\
      3.&\texttt{\{\textbf{if} ($G_i$ is heavy)                                                                                   }\\
      4.&\texttt{\ \{\textbf{if} possible,                                                                                        }\\ 
      5.&\texttt{\ \ \ assign $G_i$ to a dedicated cluster as a gang task                                                         }\\   
      ~ &\texttt{\ \ \ that encapsulates a flattened schedule;                                                                    }\\   
      6.&\texttt{\ \ \textbf{else}                                                                                                }\\ 
      7.&\texttt{\ \ \ skip its assignment for now;                                                                               }\\ 
      8.&\texttt{\ \}                                                                                                             }\\ 
      9.&\texttt{\ \textbf{else} //$G_i$ is light                                                                                 }\\
     10.&\texttt{\ \{\textbf{if} possible,                                                                                        }\\
     11.&\texttt{\ \ \ assign $G_i$ as a single-threaded task using First-Fit on a                                                }\\ 
      ~ &\texttt{\ \ \ separate pool of processors ("bins") that expands as needed;                                               }\\ 
     12.&\texttt{\ \ \textbf{else}                                                                                                }\\   
     13.&\texttt{\ \ \ skip its assignment for now;                                                                               }\\ 
     14.&\texttt{\ \}                                                                                                             }\\ 
     15.&\texttt{\}                                                                                                               }\\ 
      ~ &\texttt{~~                                                                                                               }\\ 
     16.&\texttt{//2nd pass -- inspired by C=D semi-partitioning                                                                  }\\ 
     17.&\texttt{\textbf{for} (every DAG $G_i$ that is still unassigned)                                                          }\\ 
     18.&\texttt{\{\textbf{if} ($G_i$ is heavy)                                                                                   }\\ 
     19.&\texttt{\ \{\textbf{if} possible,                                                                                        }\\ 
     20.&\texttt{\ \ \ split $G_i$ over already-formed clusters, as in Fig.~\ref{fig:splitting}                                   }\\ 
      ~ &\texttt{\ \ \ using C=D sensitivity analysis and flattened schedules;                                                    }\\ 
     21.&\texttt{\ \ \textbf{else}                                                                                                }\\ 
     22.&\texttt{\ \ \ \textbf{return} FAILURE;	                                                                                  }\\   
     23.&\texttt{\ \}                                                                                                             }\\ 
     24.&\texttt{\ \textbf{else} //$G_i$ is light                                                                                 }\\ 
     25.&\texttt{\ \{\textbf{if} possible,                                                                                        }\\ 
     26.&\texttt{\ \ \ split $G_i$ over as a single threaded task over multiple bins                                              }\\ 
      ~ &\texttt{\ \ \ using C=D semi-partitioning;                                                                               }\\ 
     27.&\texttt{\ \ \textbf{else} //if we run out of usable bins                                                                 }\\ 
     28.&\texttt{\ \ \{if possible,                                                                                               }\\ 
     29.&\texttt{\ \ \ \ try to further split $G_i$ over existing clusters,                                                       }\\ 
      ~ &\texttt{\ \ \ \ as parallel gang threads, analogously as with heavy DAGs;                                                }\\ 
     30.&\texttt{\ \ \ \textbf{else}                                                                                              }\\ 
     31.&\texttt{\ \ \ \ \textbf{return} FAILURE;	                                                                              }\\ 
     32.&\texttt{\ \ \}                                                                                                           }\\   
     33.&\texttt{\ \}                                                                                                             }\\ 
     34.&\texttt{\}                                                                                                               }\\ 
     35.&\texttt{//this line is reached only when all DAGs have been assigned                                                     }\\ 
     36.&\texttt{\textbf{return} SUCCESS                                                                                          }\\ 
     \end{tabular}
     } 
     \caption{High-level pseudocode of the complete algorithm}
     \label{fig:pseudocode:small}
     \end{minipage}
     \end{flushleft}
     \end{figure}

\textbf{First pass:}
The first pass  of  algorithm
(lines 2--15 in Figure~\ref{fig:pseudocode:small}, and Figure~\ref{fig:pseudocode:detailed:part1} in detail) is akin to federated scheduling, 
with a couple of differences. The DAG tasks are considered in order, and for each "heavy" DAG task (i.e., with utilisation above 1) the algorithm
attempts to assign it to a separate, newly formed cluster; whereas, for "light" DAG tasks (i.e., with utilisation up to 1), it attempts to assign
them as sequential (i.e., non-parallel) tasks using First-Fit bin-packing on other processors, scheduled under EDF. The differences of this first
pass with federated scheduling are the following:

\begin{itemize}
\item 
Whenever the algorithm cannot manage to feasibly assign a DAG task (i.e., because there are not enough hitherto unused processors
to assign a light DAG task to/to form a new cluster for a heavy DAG task), it does not declare failure. Instead, it considers the next DAG task.
The skipped DAG task will be reconsidered in the second pass of the algorithm, using the semi-partitioned assignment technique.
\item
Whenever a new cluster is formed (i.e, to assign a heavy DAG task to), the size of the cluster under formation is determined, by default, 
on the basis of the DAG in consideration being scheduled according to a segmented and flattened schedule (as described in 
Section~\ref{sect:flattening}). More specifically, the size of the cluster is calculated according to the function 
feasibly\_max\_flatten (see Figure~\ref{fig:pseudocode:segment:maximum:flattening}). By comparison, federated scheduling uses the bound of
Equation~\ref{eq:graham:cluster:size}, under the assumption of work-conserving node scheduling. Our algorithm only uses 
Equation~\ref{eq:graham:cluster:size} for cluster sizing in the very rare cases that it would lead to a smaller cluster than the bound for a 
segmented and flattened schedule. In those rare cases, such a DAG task is assigned as a gang task whose threads' WCET is given by 
Equation~\ref{eq:graham} and the node scheduling inside them is work-conserving. 
\end{itemize}

In principle, there exist several options for the order in which to consider the DAG tasks for assignment.
Without loss of generality, in this paper (and in our experiments), we consider the DAGs in order of non-increasing deadline. 
This is because this ordering works well for C=D semi-partitioning, to which the second pass of our algorithm is broadly analogous.

\textbf{Second pass:}
In its second pass (lines 17-34 in Figure~\ref{fig:pseudocode:small}, and Figure~\ref{fig:pseudocode:detailed:part2} in detail), 
the algorithm attempts to assign the DAG tasks that were left unassigned during the first pass. 
The assignments (if successful) of those DAG tasks are performed in a semi-partitioned manner, meaning that their execution will be split across two or more
assignment targets (i.e., clusters or bins). The unassigned DAGs are considered in the same order as in the first pas.

When attempting to assign a \textbf{heavy} DAG task, the algorithm considers the existing clusters one by one as potential assignment targets, 
assigning a "piece" of the execution of DAG task in each cluster, until the DAG task is fully assigned. There exist various options,
for the order in which to consider the clusters. In this paper (and in our experiments), it is by order of non-increasing 
\textit{normalised gross utilisation}, which we define as the total density of all gang tasks assigned to the cluster (encapsulating gang 
tasks or pieces thereof) divided by the cluster size. This mirrors the "fullest-to-emptiest" processor ordering that usually works well for 
semi-partitioned assignments under C=D.

When attempting to assign a \textbf{light} DAG task in semi-partitioned manner, the algorithm first considers the (already populated) bins as 
potential assignment targets, whereupon, the DAG task under assignment is scheduled as a single-threaded task (i.e., with its nodes executed sequentially).
This arrangement amounts, essentially, to the original C=D algorithm, for non-parallel tasks. However, if we run out of bins to consider (which can also happen 
halfway through the splitting of a DAG task), then the algorithm also considers the existing clusters as potential assignment targets. When a piece of a light
DAG task is assigned to a given cluster, it is scheduled there as a parallel task, as in the case of heavy DAG tasks. That is, as a gang task with the same 
degree of parallelism as the target cluster, encapsulating the corresponding segmented and flattened schedule. In our experiments, bins are considered in order
of non-increasing utilisation, and clusters are considered (if necessary) in order of non-increasing normalised gross utilisation.

The algorithm declares failure if some DAG task cannot be feasibly split. Otherwise, once all DAG tasks have been assigned, success is declared.

\subsection{Additional discussion}

Some details, in the interests of accurately documenting our algorithm, as configured for our experiments:

For the purposes of sensitivity analysis, the exact dbf-based analysis (Equation~\ref{eq:sensitivity:analysis:dbf}) can be employed. 
However, in our experiments in this paper, we opt instead for the slightly pessimistic but low-complexity analysis of Augusto, as restated for
constrained deadlines (Equation~\ref{eq:augusto:3:constrained}), in order to run the experiments quickly.

Whenever a DAG task is split according to C=D semantics, the clusters where its zero-laxity pieces (i.e., all pieces except possibly the last one)
are assigned are not considered any more as potential assignment targets of additional DAG task pieces (lines 58 and 76 in 
Figure~\ref{fig:pseudocode:detailed:part2}, resp. for heavy and light DAG tasks). This is because any uniprocessor system with two or more 
zero-laxity tasks cannot be feasible; the argument extends to a processor cluster with gang tasks with a degree of parallelism equal to the 
cluster size. Although they might theoretically still be feasible assignment targets for non-zero-laxity final pieces of DAG tasks, in practice 
there is little to be gained, so we forego this option in the interests of simplicity.

In line~54 of Figure~\ref{fig:pseudocode:detailed:part2}, there is a check of whether the response time of a heavy DAG task $G_i$ being split
has, by that point in its schedule, exceeded its deadline of $D_i$. In the original, non-parallel, C=D scheduling algorithm there is no need
for such check every time a zero-laxity piece is instantiated, because the fact that such pieces always execute at the highest-priority, suffering 
no interference, suffices to rule out deadline violations. However, with parallel tasks, even without interference, deadline violations while splitting
may result from insufficient parallelisation in the semi-partitioned schedule. For example, consider a DAG task consisting of 3 parallel nodes, all 
with a WCET of 20, with the DAG task's deadline being $20+\epsilon$. If, upon splitting this DAG task, we keep assigning its pieces to clusters of 
just two processors, eventually its deadline will be violated before its semi-partitioned schedule is completed. For light DAG tasks, there is no such
check, because (as with non-parallel C=D) even without any parallel execution, such tasks can always meet their deadlines, in the absence of interference.

     \begin{figure}[H]
     \begin{flushleft}
     \begin{minipage}[t]{8.5cm}  
     \centering
     \footnotesize{
     \begin{tabular}{rp{8cm}}
      1.&\texttt{int q=0; //counts clusters                                                                         }\\
      2.&\texttt{int r=0; //counts populated bins                                                                   }\\
      3.&\texttt{int M\_empty=m; //counts empty processors                                                          }\\
      4.&\texttt{                                                                                                   }\\ 
      5.&\texttt{//first pass: assignments without splitting                                                        }\\   
      6.&\texttt{\textbf{for} (int i=1 to n) //DAGs indexed in order of non-increasing $T_i$                        }\\ 
      7.&\texttt{\{\textbf{if} ($U_i$ >1) //heavy DAG                                                               }\\ 
      8.&\texttt{\  \{int $m^{\prime}$=cluster\_size\_requirements($G_i$);                                          }\\ 
      9.&\texttt{\ \  if ($m^{\prime}$<=M\_empty)                                                                   }\\
     10.&\texttt{\ \ \ \{q=q+1;                                                                                     }\\
     11.&\texttt{\ \ \ cluster $Q_q$=new cluster($m^{\prime}$);                                                     }\\ 
     12.&\texttt{\ \ \ assign($G_i$, $Q_q$); //as a gang task                                                       }\\   
     13.&\texttt{\ \ \ M\_empty=M\_empty-$m^{\prime}$;                                                              }\\ 
     14.&\texttt{\ \ \}                                                                                             }\\ 
     15.&\texttt{\ \ //else, leave for 2nd pass, to try and assign with splitting                                   }\\ 
     16.&\texttt{\ \}                                                                                               }\\ 
     17.&\texttt{\ \textbf{else} //$U_i$<=1; light DAG                                                              }\\ 
     18.&\texttt{\ \{boolean ff=assign\_first\_fit($G_i$, $b_1$ to $b_r$); //density-based test                     }\\
     19.&\texttt{\ \ \textbf{if} (ff==TRUE) //success                                                               }\\ 
     20.&\texttt{\ \ \ \textbf{break};                                                                              }\\ 
     21.&\texttt{\ \ \textbf{else if} (M\_empty>0)                                                                  }\\ 
     22.&\texttt{\ \ \{r=r+1;                                                                                       }\\   
     23.&\texttt{\ \ \ $b_r$=new bin(); //a bin is a single-processor cluster                                       }\\ 
     24.&\texttt{\ \ \ assign($G_i$, $b_r$); //as single-threaded task                                              }\\ 
     25.&\texttt{\ \ \ M\_empty=M\_empty-1;                                                                         }\\ 
     26.&\texttt{\ \ \}                                                                                             }\\ 
     27.&\texttt{\ \ //else leave for 2nd pass, to try and assign with splitting                                    }\\ 
     28.&\texttt{\ \}                                                                                               }\\ 
     29.&\texttt{\}                                                                                                 }\\ 
     \end{tabular}
     } 
     \caption{Detailed pseudocode - 1st pass of the  algorithm}
     \label{fig:pseudocode:detailed:part1}
     \end{minipage}
     \end{flushleft}
     \end{figure}

   \begin{figure*}[htbp]
     \begin{flushleft}
     \begin{minipage}[t]{18cm}  
     \centering
     \footnotesize{
     \begin{tabular}{rp{8.5cm}rp{8cm}}
     30.&\texttt{//second pass: assignments with splitting                                                         }&  63. &\texttt{\ \textbf{else} //$U_i$<=1; $G_i$ is light                                               }\\
     31.&\texttt{\textbf{for} (int i=1 to n, where $G_i$ is still unassigned)                                      }&  64. &\texttt{\ \{int C=$W_i$;                                                                         }\\
     32.&\texttt{\{list OQ = list of ($Q_1$,$\ldots$, $Q_q$) with $Q_q$.closed!=TRUE,                              }&  65. &\texttt{\ \ int D=$D_i$;                                                                         }\\
      ~ &\texttt{\ \ \ \ \ \ \ \ \ \ ordered by non-increasing normalised gross density;                           }&  66. &\texttt{\ \ b=Ob.next();                                                                         }\\
     33.&\texttt{\ list Ob = list of ($b_1$, ..., $b_r$) with $b_q$.closed!=TRUE,                                  }&  67. &\texttt{\ \ \textbf{while} (TRUE))                                                               }\\ 
      ~ &\texttt{\ \ \ \ \ \ \ \ \ \ ordered by non-increasing density;                                            }&  68. &\texttt{\ \ \{\textbf{if} (b!=NULL) //if eligible bins exist                                     }\\   
     34.&\texttt{\ \textbf{if} ($U_i$>1) //$G_i$ is heavy                                                          }&  69. &\texttt{\ \ \ \{\textbf{if} (uniprocessor\_EDF\_schedulable($\Gamma_b$$\cup$ task(C, D, $T_i$))) }\\ 
     35.&\texttt{\ \{DAG G=$G_i$;                                                                                  }&  70. &\texttt{\ \ \ \ \{assign(FS, Q);                                                                 }\\ 
     36.&\texttt{\ \ int D=$D_i$;                                                                                  }&  71. &\texttt{\ \ \ \ \ \textbf{break}; //$G_i$ entirely assigned                                      }\\ 
     37.&\texttt{\ \ int split\_schedule\_length=0;                                                                }&  72. &\texttt{\ \ \ \ \}                                                                               }\\
     38.&\texttt{\ \ \textbf{while} (TRUE)                                                                         }&  73. &\texttt{\ \ \ \ //splitting ensues                                                               }\\
     39.&\texttt{\ \ \{Q=OQ.next();                                                                                }&  74. &\texttt{\ \ \ \ int piece\_C = C\_equals\_D\_sensitivity\_analysis(Q);                           }\\ 
     40.&\texttt{\ \ \ \textbf{if} (Q==NULL) //run out of clusters                                                 }&  75. &\texttt{\ \ \ \ assign($G_i$, b, piece\_C); //C=D=piece\_C; T=$T_i$                              }\\   
     41.&\texttt{\ \ \ \ \textbf{return FAILURE};                                                                  }&  76. &\texttt{\ \ \ \ b.closed = TRUE; //no more assignments there;                                    }\\ 
     42.&\texttt{\ \ \ schedule FS=new schedule();                                                                 }&  77. &\texttt{\ \ \ \ C = C-piece\_C;                                                                  }\\ 
     43.&\texttt{\ \ \ int $m^{\prime}$=Q.number\_of\_processors();                                                }&  78. &\texttt{\ \ \ \ D = D-piece\_C;                                                                  }\\ 
     44.&\texttt{\ \ \ \textbf{for} (segment $S_k$ in G, in order of ascending k)                                  }&  79. &\texttt{\ \ \ \}                                                                                 }\\ 
     45.&\texttt{\ \ \ \ FS.append(flatten\_segment($S_k$), $m^{\prime}$);                                         }&  80. &\texttt{\ \ \ \textbf{else} //try to assign as parallel task on some cluster;                    }\\ 
     46.&\texttt{\ \ \ //try to avoid further splitting of 2nd, 3rd etc piece                                      }&  81. &\texttt{\ \ \ \{Analogously as with heavy DAGs over lines 35--62;                                }\\
     47.&\texttt{\ \ \ \textbf{if} (uniprocessor\_EDF\_schedulable($\Gamma_G$$\cup$task(FS.length, D, $T_i$)))     }&  82. &\texttt{\ \ \ \}                                                                                 }\\ 
     48.&\texttt{\ \ \ \{assign(FS, Q);                                                                            }&  83. &\texttt{\ \ \}                                                                                   }\\ 
     49.&\texttt{\ \ \ \ \textbf{break}; //$G_i$ entirely assigned                                                 }&  84. &\texttt{\ \}                                                                                     }\\ 
     50.&\texttt{\ \ \ \}                                                                                          }&  85. &\texttt{\}                                                                                       }\\   
     51.&\texttt{\ \ \ //splitting ensues; new zero-laxity piece formed below                                      }&  86. &\texttt{//This line is only reached if all the DAGs have been assigned.                          }\\ 
     52.&\texttt{\ \ \ int piece\_C = C\_equals\_D\_sensitivity\_analysis(Q);                                      }&  87. &\texttt{\textbf{return} SUCCESS;                                                                 }\\ 
     53.&\texttt{\ \ \ split\_schedule\_length=split\_schedule\_length+piece\_C;                                   }&      &\texttt{                                                                                         }\\ 
     54.&\texttt{\ \ \ \textbf{if} (split\_schedule\_length>=$D_i$)\,//deadline reached w/o completion?            }&      &\texttt{                                                                                         }\\ 
     55.&\texttt{\ \ \ \ \textbf{return} FAILURE;                                                                  }&      &\texttt{                                                                                         }\\ 
     56.&\texttt{\ \ \ piece\_FS=FS.interval(0, piece\_C);\,//first piece\_C time units of FS                      }&      &\texttt{                                                                                         }\\ 
     57.&\texttt{\ \ \ assign(piece\_FS, Q); //with deadline=piece\_C and period $T_i$                             }&      &\texttt{                                                                                         }\\ 
     58.&\texttt{\ \ \ Q.closed = TRUE; //no more assignments there;                                               }&      &\texttt{                                                                                         }\\ 
     59.&\texttt{\ \ \ G = leftover\_DAG(G,piece\_C); //rump DAG remaining after removal                           }&      &\texttt{                                                                                         }\\ 
      ~.&\texttt{\ \ \ \ \ \ \ \ \ \ \ \ \ \ \ \ \ \ \ \ \ \ \ \ \ \ \ \ \ \ \ \ //of execution in piece\_FS       }&      &\texttt{                                                                                         }\\ 
     60.&\texttt{\ \ \ D = D-piece\_C;               //rump DAG's deadline; the period is $T_i$                    }&      &\texttt{                                                                                         }\\   
     61.&\texttt{\ \ \}                                                                                            }&      &\texttt{                                                                                         }\\ 
     62.&\texttt{\ \}                                                                                              }&      &\texttt{                                                                                         }\\ 
     \end{tabular}
     } 
     \caption{Detailed pseudocode - 2nd pass of the  algorithm}
     \label{fig:pseudocode:detailed:part2}
     \end{minipage}
     \end{flushleft}
     \end{figure*}

\subsection{Conceptual comparisons with other work}

Conceptually, there are some analogies between our approach and the bundled gang scheduling approach by 
Wasly and Pellizzoni~\cite{Wasly_Pellizzoni_19}, but also significant differences.

The segments into which our approach breaks up a DAG are broadly analogous to Wasly's "bundles", in that, regardless of the algorithm for 
forming segments and bundles, in both cases they are populated by nodes that can (in our approach) or \textit{should} (in Wasly's approach) 
be executed in parallel to each other. Another similarity is that, under both approaches, the degree of parallelism with which a task is 
scheduled may vary at different points during its execution. However, whereas, under bundled gang scheduling, each node in a bundle gets its own
dedicated thread, under Segmented-Flattened-and-Split scheduling this is typically not the case. Instead, McNaughton's algorithm is used, to match
the number of threads to the number of processors in the target cluster. Another difference is that the degree of parallelism (thread count), under 
bundled gang scheduling, is (only) modified at the bundle boundaries; whereas, under our approach, the parallel schedule can be split at 
arbitrary points and resumed, possibly with a different degree of parallelism, on a different cluster. The migration points in the schedule 
can then be selected according to sensitivity analysis, in order to utilise as much as possible of a given cluster's available processing capacity.

Both our approach and that of Wasly and Pellizzoni make use of the gang scheduling abstraction in their formulations. However, in our case, 
the underlying scheduling analysis, as in C=D semi-partitioning, is that for uniprocessor EDF. This is because, on each cluster, all DAG tasks assigned 
there (either in their entirety or as pieces of DAG tasks that are split) have the same number of threads, that matches the size of the cluster.
This simplifies the analysis, eliminates the possibility of 2D-blocking, and also allows for independently testing for the feasibility of each cluster. 
Conversely, Wasly and Pellizzoni genuinely use a gang scheduling policy, and this carries on to their analysis.

In any case, it is important to stress that Wasly and Pellizzoni's scheduling algorithm design was primarily motivated by parallel applications 
for which it is of utmost importance for certain subtasks (nodes) to be executed in lockstep, or else severe performance penalties can arise, e.g., 
due to the need for synchronisation. In such a context, their design choices make good sense, and, e.g., the use of McNaughton's wraparound algorithm 
would have been counter-productive. In contrast, our work makes no such assumption about the application; we merely consider the DAG task deadlines and 
the precedence constraints among nodes and try to optimise our algorithm for scheduling performance under those constraints. 
This is important context for the scope of any comparison.

\section{Evaluation}
In this section, we provide evaluation results for our proposed scheduling algorithm using randomly generated synthetic DAG tasksets. 
By applying the respective schedulability tests to the synthetic workloads, we compare our proposed scheduling algorithm
with federated scheduling~\cite{Li_CALGS_14} in terms of scheduling success ratio.

\subsection{Experimental setup}
For our experiments, we created synthetic workloads using the random DAG generator~\cite{Dai_22} by Xiaotian Dai from the University of York.
The random DAG task set generation process has two parts. First, the "macroscopic" attributes of each DAG task in the set, i.e., 
its utilisation, its interarrival time and its deadline, are generated. Subsequently, the tool randomly generates, for each DAG task, 
a number of nodes and edges among (some of) them.

For the first part, the generator uses the UUnifast-Discard method \cite{Davis_IB_09}. It requires two inputs: the total utilisation of the DAG task set 
to be generated and the number of DAG tasks within it. Based on those, it randomly generates uniformly distributed utilisations for the different 
DAG tasks. Once the utilisation ($U_i$) of each DAG task is determined, its interarrival time ($T_i$) is randomly chosen from a predefined list 
of possible period values; the defaults for those, also used in our experiments are: 100, 200, 500, 1000, 2000 and 5000. 
From the utilisation and the interarrival time of each DAG task, its workload is then computed as $W_i = T_i \cdot U_i$. In our experiments, 
we considered implicit-deadline DAG tasks, so the deadline of each DAG task is equal to its interarrival time. 

Once the above attributes are defined for each DAG, for the second part of the DAG task generation process,
the internal structure of each DAG is progressively generated in \textit{layers}. 
A layer of a DAG is defined as the collection of nodes that share the same \textit{minimum} hop distance from the source node, over all 
paths. (Note that this is different to our definition of a segment, in Section~\ref{sect:segmentation}, which is based on a node's 
\textit{maximum} hop distance from the source node, over all paths.)
The total number of layers in a DAG is randomly chosen from 4 to 10. The number of generated nodes in each layer is uniformly distributed 
from 2 to 5. The DAG generation tool randomly adds edges between nodes; we went with the default settings/probabilities for that~\cite{Zhao_DB_22}.
The already determined workload $W_i$ of each DAG is then distributed randomly amongst the nodes generated. 
We slightly modified the code to ensure that "dummy" source and sink nodes, with a WCET of 0 are added to the generated DAG, for our convenience.

Our experiments considered platforms with $m=8$ or $16$ identical processors, and, in each case, 
with $n=10$ or $20$ DAG tasks per DAG task set.
In the plots, the horizontal axis ($U$) represents the DAG task set's normalised utilisation, i.e., the fraction of the platform capacity that 
the DAG task set utilisation amounts to. For example, for a platform with 8 processors, a data point with $U=85\%$ represents DAG task sets 
that have a utilisation of $8 \cdot 0.85=6.8$. For every value of $U$ from 5\% to 100\%, with a step size of 5\%, we generated 100 DAG task sets 
for each experiment/combination of $m$ and $n$. The vertical axis represents the fraction of task sets of a given utilisation that were found schedulable
by the respective test.
Table~\ref{tab:experimental:setup} summarises the combinations of parameters in our experiments. 

\begin{table}[H]
\begin{center}
\begin{small}
\begin{tabular}{|c|l|l|}
\hline
$m$ & number of processors                 & \{8, 16\}                  \\ 
\hline
$n$ & number of DAG tasks per set          & \{10, 20\}                 \\  
\hline
$U$ & normalised DAG task set utilisation  & 5\% to 100\%, in 5\% steps \\  
\hline
\end{tabular}
\caption{\label{tab:experimental:setup}Parameter combinations used in the experiments.}
\end{small}
\end{center}
\end{table} 

\subsection{Results}
Figure \ref{fig:results} shows the comparison between our Segmented-Flattened-and-Split scheduling approach (denoted "SFS" in the figure) 
and Federated Scheduling approach (denoted by "FS"), when there are 10 DAGs in each DAG task set.
For each DAG task set, we apply the schedulability test for Segmented-Flattened-and-Split scheduling and for Federated Scheduling 
to determine its schedulability under each approach. In our experiments, Federated Scheduling uses First-Fit bin-packing for light tasks,
in conjunction with a utilisation-based test on each processor, which is exact, because the DAG tasks are implicit-deadline.

\begin{figure}[t]
    \centerline{\includegraphics[scale=0.32, trim=1.5cm 0cm 1cm 1.5cm, clip]{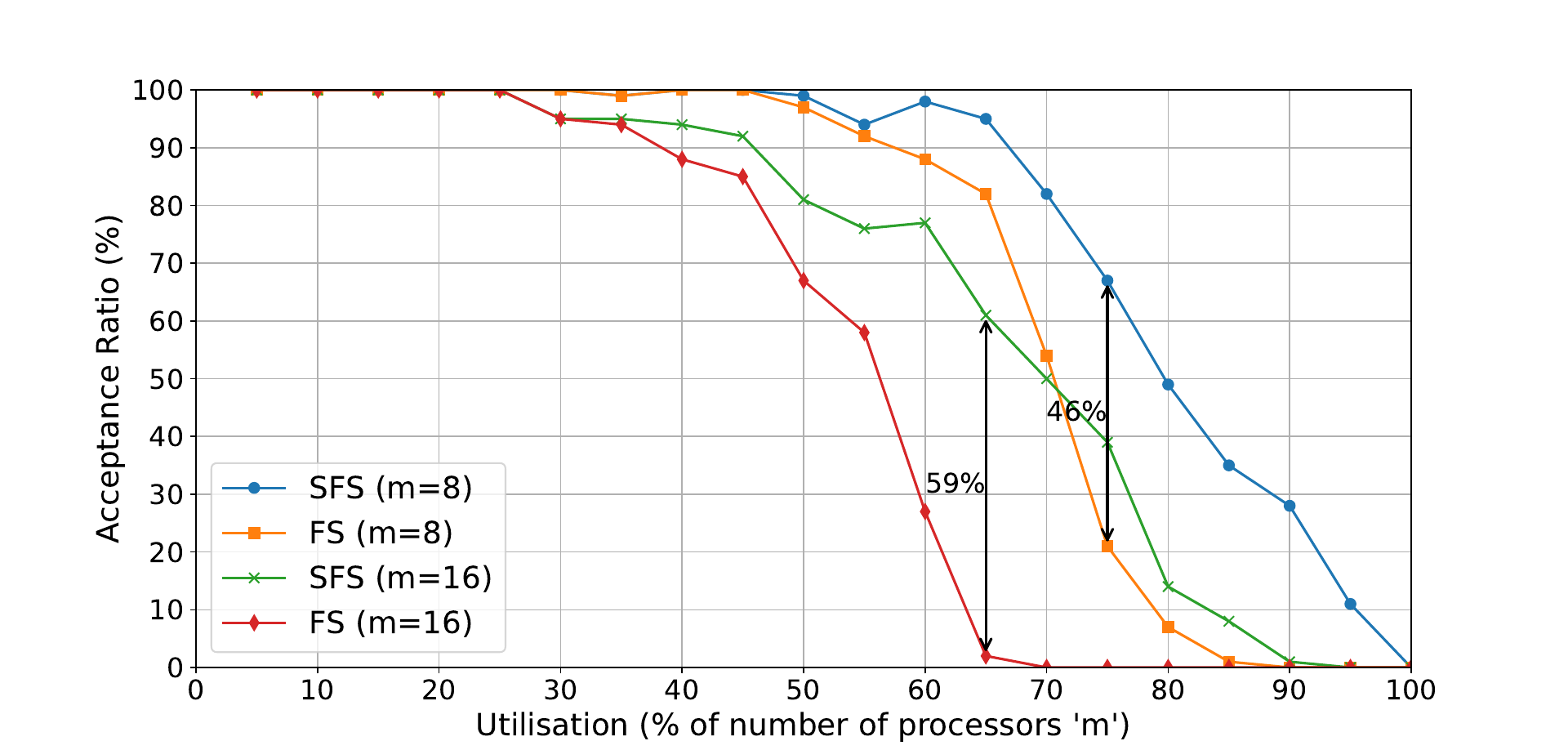}}
    \caption{Evaluation Results: Segmented-Flattened-and-Split vs. Federated Scheduling for 10 DAGs in a task set}
    \label{fig:results}
\end{figure}

\begin{figure}[t]
    \centerline{\includegraphics[scale=0.32, trim=1.5cm 0cm 1cm 1.5cm, clip]{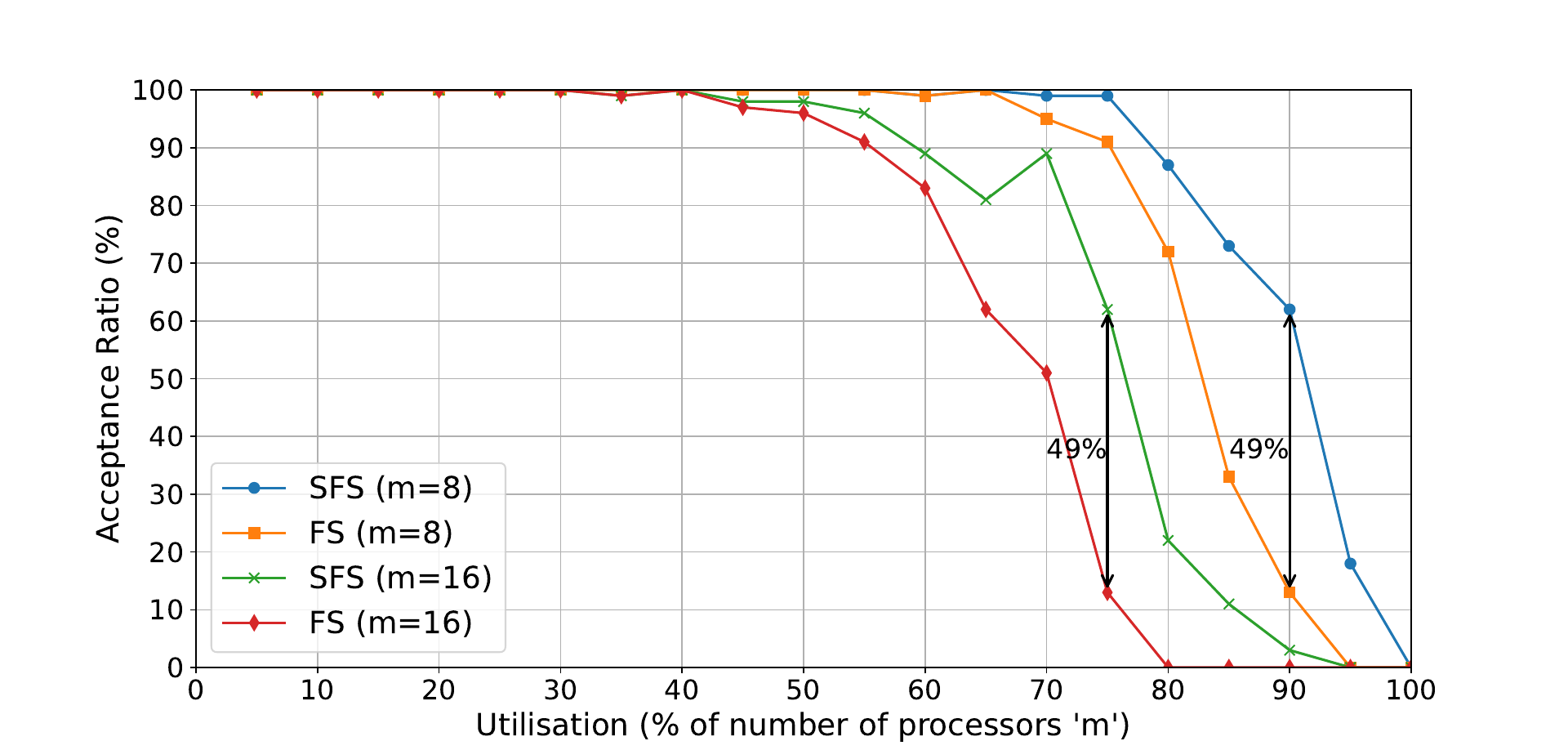}}
    \caption{Evaluation Results: Segmented-Flattened-and-Split vs. Federated Scheduling for 20 DAGs in a task set}
    \label{fig:results1} 
\end{figure}

Our Segmented-Flattened-and-Split scheduling provides a higher acceptance ratio than Federated Scheduling, both for $m=8$ and for $m=16$ processors.
The maximum vertical distance between the acceptance ratio curves of the two approaches is 46\% and 59\%, for 8 and 16 processors respectively, 
as shown in Figure \ref{fig:results}.

One can also see that the acceptance ratio of both scheduling algorithms in these experiments decreases significantly as the number of processors 
increases. The reason is that increasing the number of cores while keeping both the normalised utilisation of the DAG task set and the number of 
DAG tasks in it fixed, during the synthetic workload generation process, leads to an increase in the average DAG task utilisation. This, in turn, 
makes the number of heavy DAG tasks in the generated DAG task set much higher on average. 
To illustrate, when the normalised utilisation is 70\%
the average DAG task utilisation in Figure~\ref{fig:results} is $0.7 \cdot 8 / 10 = 0.56$ when $m=8$ and $0.7 \cdot 16 / 10 = 1.12$ when $m=16$.
However, the average number of heavy DAG tasks per DAG task set is 1.72 (out of 10) when $m=8$ and 5.84 when $m=16$ -- a supralinear increase.
And it is precisely those DAG tasks that are hard to schedule efficiently, in terms of utilisation of processing capacity, compared to light DAG tasks
for which efficient classic algorithms for non-parallel tasks can be used. As seen in Figure~\ref{fig:results}, 
Segmented-Flattened-and-Split scheduling is considerably less sensitive to a higher fraction of heavy DAG tasks, than Federated Scheduling.
This is the combined effect of both the use of flattened schedules and the semi-partitioning.

Figure \ref{fig:results1} shows the results of the experiments when the number of DAGs in each DAG task set is increased to 20.
Our algorithm once again consistently outperforms Federated Scheduling and the same observations apply. The maximum vertical distance between 
the acceptance ratio curves is 49\%, both for 8 and for 16 processors. Again, for the reasons explained earlier, both algorithms perform worse 
when the number of processors is higher, all other things remaining equal.

\section{Conclusion}
In this paper, we proposed a new algorithm (Segmented-Flattened-and-Split scheduling) for efficiently scheduling parallel DAG tasks on identical multiprocessors.
The building blocks of our algorithm include the \textit{(i)~segmentation} of DAG into disjoint sub-DAGs which can be executed in sequence. 
Then, \textit{(ii)~flattening} the segmented DAG using McNaughton's wraparound rule. 
Finally, \textit{(iii)~splitting} off a flattened schedule, inspired by C=D semi partitioning, over disjoint processor clusters.

Our algorithm assigns DAG tasks in two passes. The first pass is analogous to federated scheduling with a couple of differences, namely 
that cluster size requirements are determined on the basis of a flattened schedule and that, whenever a DAG task cannot be assigned, 
it is skipped, in order to be considered in the second pass. In the second pass, a technique inspired by C=D semi-partitioning
is used, to split the execution of unassigned DAG tasks onto different existing clusters, of potentially different sizes.

In our experiments with synthetic DAG task sets, we compared the scheduling performance of the proposed algorithm, in terms of acceptance ratio, 
with that of federated scheduling. The results show a consistently significant performance advantage. We expect that, e.g, with further refinement
of the DAG task assignment heuristics (i.e., the order in which tasks and clusters are being considered) and with the use of exact EDF schedulability 
analysis, the performance potential of our approach can be even greater, and would like to explore that in future work. We also intend to adapt
Segmented-Flattened-and-Split scheduling for the case of heterogeneous multiprocessor scheduling with unrelated processor types. 
 

\bibliographystyle{ACM-Reference-Format}
\bibliography{references}

\end{document}